\documentclass[journal]{IEEEtran}
\usepackage{amsmath,amsfonts}
\usepackage{algorithmic}
\usepackage{algorithm}
\usepackage{array}
\usepackage{textcomp}
\usepackage{stfloats}
\usepackage{url}
\usepackage{verbatim}
\usepackage{graphicx}
\usepackage{cite}

\usepackage{amssymb}
\usepackage{amsthm}
\usepackage{bm}
\usepackage{subfigure}
\usepackage{titlesec} 

\usepackage{color}

\newtheoremstyle{mytheorem}  
  {4pt}                      
  {2pt}                      
  {\itshape}                 
  {}                         
  {\bfseries}                
  {.}                        
  { }                        
  {}                         
\theoremstyle{mytheorem}     

\newtheorem{prop}{Proposition}

\makeatletter
\makeatletter
\renewenvironment{proof}[1][\proofname]{\par
  \pushQED{\qed}%
  \normalfont \topsep2pt \partopsep4pt 
  \trivlist
  \item[\hskip\labelsep
        \itshape
    #1\@addpunct{.}]\ignorespaces
}{%
  \popQED\endtrivlist\@endpefalse
}

\titlespacing{\section}{0pt}{2.5ex plus .0ex minus .0ex}{.3ex plus .0ex} 
\titlespacing{\subsection}{0pt}{2ex plus .0ex minus .0ex}{.3ex plus .0ex} 

\allowdisplaybreaks[4]

\begin{document}
\bstctlcite{BSTcontrol}

\title{\huge Degrees of Freedom of Holographic MIMO \\ in Multi-user Near-field Channels}

\author{
Houfeng~Chen,~\IEEEmembership{Student~Member,~IEEE,} 
Shaohua~Yue,~\IEEEmembership{Student~Member,~IEEE,} \\
Marco~Di~Renzo,~\IEEEmembership{Fellow,~IEEE,}
and~Hongliang~Zhang,~\IEEEmembership{Member,~IEEE} 

\vspace{-2.5em}
\thanks{Houfeng Chen, Shaohua Yue and Hongliang Zhang are with the State Key Laboratory of Advanced Optical Communication Systems and Networks, School of Electronics, Peking University, Beijing 100871, China (e-mail: \{houfengchen, yueshaohua\}@pku.edu.cn; hongliang.zhang92@gmail.com).}
\thanks{M. Di Renzo is with Universit\'e Paris-Saclay, CNRS, CentraleSup\'elec, Laboratoire des Signaux et Syst\`emes, 3 Rue Joliot-Curie, 91192 Gif-sur-Yvette, France. (marco.di-renzo@universite-paris-saclay.fr), and with King's College London, Centre for Telecommunications Research -- Department of Engineering, WC2R 2LS London, United Kingdom (marco.di\_renzo@kcl.ac.uk).}}





\maketitle

\begin{abstract}
Holographic multiple-input multiple-output (HMIMO) is an emerging technology for 6G communications, in which numerous antenna units are integrated in a limited space. As the HMIMO array aperture expands, the near-field region of the array is dramatically enlarged, resulting in more users being located in the near-field region. This creates new opportunities for wireless communications. In this context, the evaluation of the spatial degrees of freedom (DoF) of HMIMO multi-user systems in near-field channels is an open problem, as the analytical methods utilized for evaluating the DoF in far-field channels cannot be directly applied.
In this paper, we propose a novel method to calculate the DoF of HMIMO in multi-user near-field channels. We first derive the DoF for a single user in the near field, and then extend the analysis to multi-user scenarios. In this latter scenario, we focus on the impact of spatial blocking between HMIMO users. The derived analytical framework reveals that the DoF of HMIMO in multi-user near-field channels is not in general given by the sum of the DoF of the HMIMO single-user setting because of spatial blocking. 
Simulation results demonstrate the effectiveness of the proposed method. In the considered case study, the number of DoF reduces by $21.2\%$ on average due to spatial blocking.
\end{abstract}


\section{Introduction}\label{Intro}
\IEEEPARstart{H}{olographic} multiple-input multiple-output (HMIMO) is an emerging technology for 6G communications, which integrates a large number of sub-wavelength antenna units on a finite-size surface~\cite{gong2023holographic,deng2022reconfigurable,elmossallamy2020reconfigurable}. As a result, HMIMO arrays significantly increase the number of degrees of freedom (DoF), which corresponds to the number of orthogonal sub-channels allowed by the communication channel. 
As the aperture of the HMIMO array increases, the near-field region of the antenna array is extended, which results in an increased number of users located in the near field of the HMIMO array. In the near-field region, the eletromagnetic waves are characterized by spherical wavefronts, in contrast to the far-field region in which the wavefront is planar.
The planar wavefront of far-field waves, which significantly simplifies the expression of the radiation kernels in far-field scenarios, is not sufficiently accurate to model spherical waves.
Therefore, the methods for the analysis of the DoF in far-field channels cannot be applied to the analysis of the DoF in near-field channels.

Most of the existing works on the DoF of MIMO systems are applicable to far-field channels~\cite{song2024line,sridharanDegreesFreedomMIMO2015}.
In \cite{song2024line}, a line-of-sight (LoS) MIMO system is analyzed based on the planar wave assumption and the number of DoF is analyzed.
The authors of \cite{sridharanDegreesFreedomMIMO2015} investigate the DoF of MIMO networks based on decomposition and linear beamforming.
As mentioned, however, these results cannot be directly applied to near-field communications. 
The analysis of the DoF in near-field free-space channels has been a subject of intense research activities for several years \cite{bucciDegreesFreedomScattered1989, millerCommunicatingWavesBetween2000,franceschettiDegreesFreedomWireless2011}. 
Thanks to the introduction of HMIMO, the analysis of the DoF in near-field channels has recently attracted renewed interest~\cite{dardariCommunicatingLargeIntelligent2020,pizzoLandauEigenvalueTheorem2022,Renzo2023}.
Two main methods exist for the analysis of the DoF in near-field channels: (1) the spatial bandwidth \cite{bucciDegreesFreedomScattered1989} and (2) Landau's eigenvalue theorem \cite{pizzoLandauEigenvalueTheorem2022}. These methods have been recently applied, e.g., in \cite{dardariCommunicatingLargeIntelligent2020}, \cite{dingDegreesFreedom3D2022} and \cite{Renzo2023}, respectively. 

However, the DoF analysis of near-field multi-user communication systems is not available in the literature.
The presence of multiple users equipped with HMIMO arrays brings new challenges when computing the number of DoF. Specifically, users equipped with HMIMO arrays may block the signal that other users equipped with HMIMO arrays receive from a base station (BS). In this case, the number of DoF of the multi-user channel may not be equal to the sum of the DoF of the users when they are the only ones available in the network, as the spatial blocking may reduce the number of DoF.

Motivated by these considerations, we introduce an analytical framework to compute the number of DoF in multi-user HMIMO near-field channels when the links from the BS to each user may be blocked by the presence of other users. The proposed approach for estimating the number of DoF is based on the spatial bandwidth method, which was introduced in \cite{BucciRepresentationElectromagneticFields1998, dardariCommunicatingLargeIntelligent2020}. We identify the conditions under which the number of DoF in multi-user networks is and is not equal to the sum of the DoF in single-user networks, and in the latter case the reduction in the number of DoF is quantified analytically. 

\section{System Model}\label{sysmod}
In this section, we first introduce the communication scenario for the considered HMIMO system. Then, we introduce the LoS near-field channel model. Finally, we review existing works on the attainable number of DoF in LoS channels for HMIMO-aided communications.

\subsection{Scenario Description}
As shown in Fig.~\ref{figure_model}, we consider an uplink HMIMO communication system where a BS is equipped with a large HMIMO surface, denoted by $\mathcal{R}$, for multi-user communications. Each user is equipped with a small HMIMO surface, denoted by $\mathcal{S}$, of the same size. Each HMIMO surface consists of a dense integration of numerous antenna units in a limited space, which can be approximated as a continuous surface. All the users are assumed to be located in the near-field region of the BS. Also, a LoS channel model for the BS-users links is considered.

\begin{figure}[!t]
  \centering
  \includegraphics[width=0.36\textwidth]{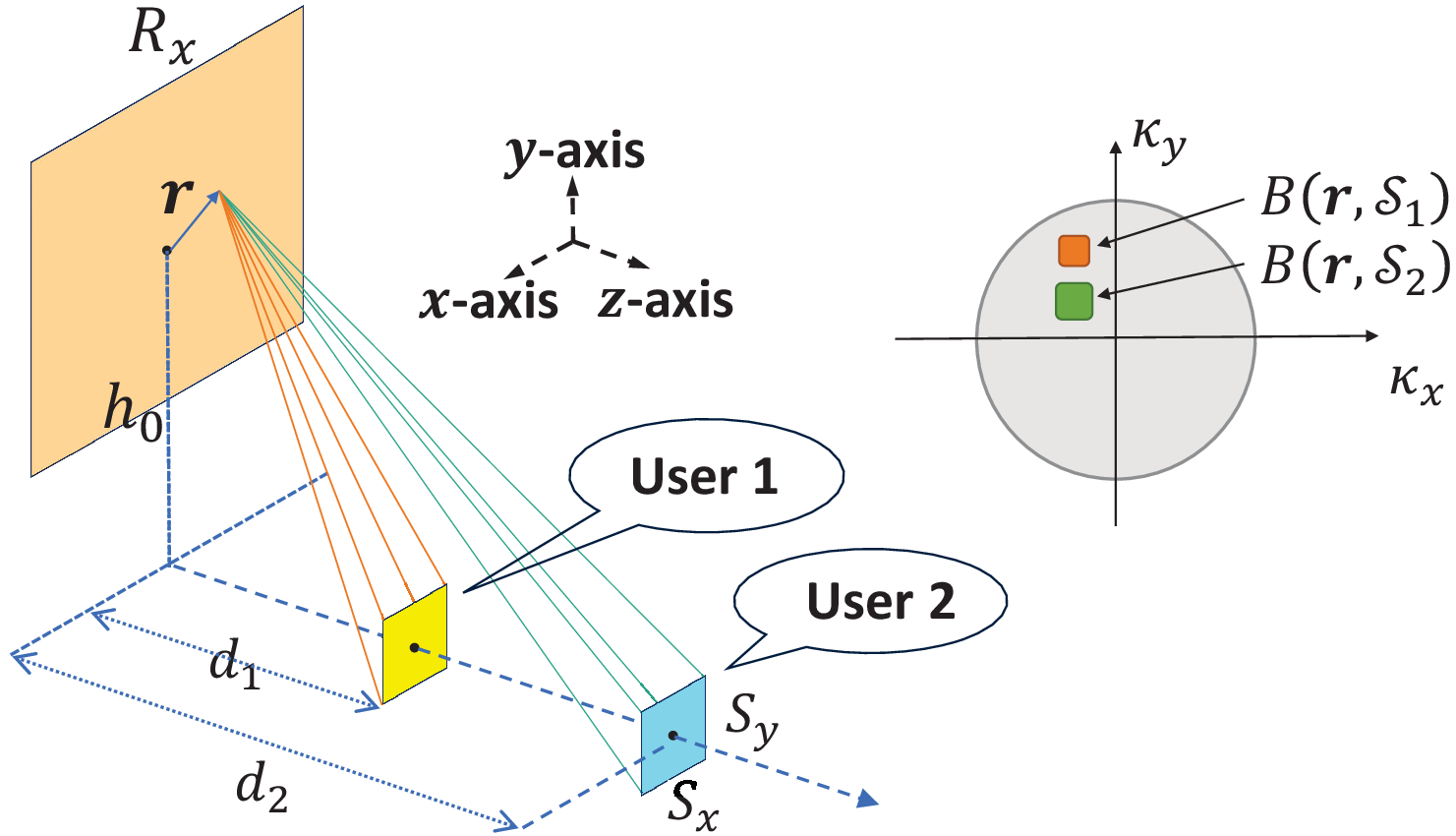}
  \caption{Considered HMIMO multi-user near-field system model.}
  \label{figure_model}
  \vspace{-1em}
\end{figure}

\subsection{LoS Channel Model}
Consider a generic user communicating with the BS using, for ease of analysis, scalar electromagnetic waves. 
The scalar electromagnetic field $E(\bm{r}), \bm{r}\in\mathcal{R}$, is the image of the current density $J(\bm{s}),\bm{s}\in\mathcal{S}$, which can be expressed through the linear channel operator $\mathcal{G}$ as \cite{pizzoLandauEigenvalueTheorem2022}
\begin{equation}
    \label{channel model}
    E(\bm{r})=(\mathcal{G}J)(\bm{r})=\int_{\mathcal{S}} h(\bm{r},\bm{s})J(\bm{s}){\rm d} \bm{s},
\end{equation}
where $h(\bm{r}, \bm{s})$ is the space-variant kernel induced by the operator. In LoS channels, it is determined by the inhomogeneous Helmholtz wave equation
and can be expressed as~\cite{dardariCommunicatingLargeIntelligent2020}
\begin{align}
    \label{Green_function}
    h(\bm{r}, \bm{s}) = -\mathrm{j} \kappa_0 \eta \frac{e^{\mathrm{j}\kappa_0 \left\Vert \bm{r}-\bm{s} \right\Vert}}{4\pi \left\Vert \bm{r}-\bm{s} \right\Vert},
\end{align}
The model in (\ref{Green_function}) is known as the non-uniform spherical
wave (NUSW) near-field channel model \cite{dingDegreesFreedom3D2022}. In the considered communication system, the electric field $E(\bm{r})$ at the observation point $\bm{r} \in \mathcal{R}$ is assumed to be measured perfectly~\cite{dingDegreesFreedom3D2022}, and the transmission of information between $\mathcal{S}$ and $\mathcal{R}$ is executed through the mapping
\begin{align}
    \label{mapping}
    J(\bm{s}), \bm{s} \in \mathcal{S}, \longrightarrow E(\bm{r}), \bm{r} \in \mathcal{R},
\end{align}%
where $J(\bm{s})$ and $E(\bm{r})$ denote the transmitted and received messages, respectively.

\subsection{Spatial Degrees of Freedom}
The number of spatial DoF available in a LoS communication channel between $\mathcal{S}$ and $\mathcal{R}$ is determined by the number of orthogonal sub-channels, allowed by the mapping in (\ref{mapping}), from the source signal space to the receiving signal space through the linear channel operator $\mathcal{G}$~\cite{Renzo2023}. Each sub-channel corresponds to an encoding function in the source signal space and a decoding function in the receiving signal space. These two sets of functions can be obtained by solving a pair of coupled eigenfunction problems, as detailed in, e.g., \cite{dardariCommunicatingLargeIntelligent2020}.
Generally, there exist an infinite number of eigenfunction pairs. However, due to the inherent band-limiting properties \cite{bucciSpatialBandwidthScattered1987} of the mapping in (\ref{mapping}), many eigenvalues are approximately equal to zero, and thus the corresponding eigenfunction pairs is not suitable for communication. If we sort all the eigenvalues in a decreasing order, there exists a transition zone where the eigenvalues rapidly tend to zero. 

Therefore, only a limited number of eigenvalues is significant, which allows for a precise definition of DoF: Sorting the eigenvalues in decreasing order as $\sigma^2_1 \geq \sigma^2_2 \geq \cdots \geq \sigma^2_n \geq \cdots$, the number of DoF, denoted as $N_\epsilon$, given a specified level of representation accuracy $\epsilon$, is defined as the index for which $\sigma^2_{N_{\epsilon}-1} > \epsilon$ and $\sigma^2_{N_\epsilon} \leq \epsilon$, which ensures that any electric field can be represented using the first $N_\epsilon$ eigenfunctions, as the signal reconstruction error is less than or equal to $\epsilon$. The number $N_\epsilon$ is referred to as the number of \textit{effective} DoF.

\subsection{Spatial Frequency and Spatial Bandwidth}\label{subsection_spatial_bandwidth}
In this paper, we compute the number of effective DoF in near-field channels based on the spatial bandwidth method\cite{bucciDegreesFreedomScattered1989,dardariCommunicatingLargeIntelligent2020,bucciSpatialBandwidthScattered1987}. More specifically, we utilize the approximation introduced in \cite{dardariCommunicatingLargeIntelligent2020}.
To elaborate, given a 2D HMIMO array, the effective DoF can be approximated by~\cite{dardariCommunicatingLargeIntelligent2020}
\begin{align}
    \label{DoF}
    N=\frac{1}{(2\pi)^2}\int_\mathcal{R} B(\bm{r}, \mathcal{S}) {\rm d}\bm{r},
\end{align}
where $B(\bm{r}, \mathcal{S})$ is the \textit{local spatial bandwidth} 
of the electric field radiated by $\mathcal{S}$ and measured at $\bm{r} \in \mathcal{R}$. A justification for (\ref{DoF}) can be found in \cite[Sec. III-B]{franceschettiDegreesFreedomWireless2011}.

To formulate the spatial bandwidth, we need to first introduce the concept of \textit{local spatial frequency}. To this end, we utilize the definitions and analytical development summarized in \cite{dingDegreesFreedom3D2022}. 
Specifically, assuming that the observation position $\bm{r}$ moves a small step ${\rm d} \bm{p} = ({\rm d} x, {\rm d} y)$ along the receiving surface $\mathcal{R}$, such that the incident directions of all source points on the transmitting surface $\mathcal{S}$ can be regarded unchanged, the phase of the wave component generated by the point source $\bm{s}$, i.e., $h(\bm{r}, \bm{s})J (\bm{s}) {\rm d} \bm{s}$, is changed by
\begin{align}
    \label{kappa_first}
    \bm{\kappa} (\bm{r}, \bm{s}) \cdot {\rm d} \bm{p} = \kappa_0  \left(\bm{\hat{v}}-\bm{\hat{n}}(\bm{\hat{v}}\cdot \bm{\hat{n}})\right) \cdot {\rm d} \bm{p}
\end{align}
where $\bm{\hat{v}} = \frac{\bm{r} - \bm{s}}{\left\Vert \bm{r} - \bm{s} \right\Vert}$, $\bm{\hat{n}}$ is the unit vector perpendicular to the surface at the point $\bm{r}$, and $(\cdot)$ denotes the scalar product between vectors. The right-hand side of (\ref{kappa_first}) is termed as the \textit{spatial frequency} of the integrand in (\ref{channel model}) \cite{dingDegreesFreedom3D2022}.
From (\ref{kappa_first}), assuming that the transmitting and receiving HMIMO surfaces lie on the $xy$-plane (as shown in Fig. \ref{figure_model}), we have $\hat{\bm{n}} = \hat{\bm{z}}$ and
\begin{align}
    \label{spatial_frequency}
    \bm{\kappa} (\bm{r}, \bm{s}) = \kappa_0  (\bm{\hat{v}}-\bm{\hat{n}}(\bm{\hat{v}}\cdot \bm{\hat{n}})) = \left( \kappa_x \left(\bm{r}, \bm{s}\right), \kappa_y \left(\bm{r}, \bm{s}\right) \right).
\end{align}%
Therefore, for   $\bm{r} \!=\! (r_x, r_y, 0)$ and $\bm{s} \!=\! (s_x, s_y, s_z)$, we get \cite{BucciRepresentationElectromagneticFields1998}
\begin{subequations}
    \vspace{-0.8em}
    \label{spatial_frequency_x_y}
    \begin{align}
        \kappa_x \left(\bm{r}, \bm{s}\right) = \kappa_0 \frac{r_x-s_x}{\sqrt{(r_x-s_x)^2 + (r_y-s_y)^2 + s_z^2}}  \\
        \kappa_y \left(\bm{r}, \bm{s}\right) = \kappa_0 \frac{r_y-s_y}{\sqrt{(r_x-s_x)^2 + (r_y-s_y)^2 + s_z^2}}.
    \end{align}
\end{subequations}
The \textit{local spatial bandwidth} in the wavenumber domain 
(i.e., in the $\kappa_x\!-\!\kappa_y$ plane) evaluated at the point $\bm{r}$ on $\mathcal{R}$ is defined as the maximum wavenumber spread evaluated across all the point sources on $\mathcal{S}$~\cite{dardariCommunicatingLargeIntelligent2020}. It can be expressed as 
\begin{align}
    \label{local_spatial_bandwidth}
        B(\bm{r}, \mathcal{S})={\rm m}\left(\bm{\kappa} (\bm{r}, \bm{s})\right)_{\bm{s} \in \mathcal{S}}, 
\end{align}
where ${\rm m}(\cdot)_{\bm{s} \in \mathcal{S}}$ is the Lebesgue measure of the region in the wavenumber domain spanned by $\bm{\kappa}(\bm{r}\!, \bm{s})$ when $\bm{s}$ varies on $\mathcal{S}$.

\section{Number of Effective DoF} \label{section3}
In this section, we utilize (\ref{spatial_frequency_x_y}) and (\ref{local_spatial_bandwidth}) to compute the number of effective DoF in single and multiple user settings. As for the single user setting, we generalize the framework in \cite{dardariCommunicatingLargeIntelligent2020} by considering that the center-points of the user and BS HMIMO surfaces on the $xy$-plane do not necessarily coincide with one another. As for the multi-user case, no analytical expressions exist in the open technical literature.

\subsection{Single User Setup}
Without loss of generality, we assume that the HMIMO surfaces of the BS and user are parallel, with a height difference $h_0$ between their center points along the $y$-axis. $S_x$, $S_y$, and $R_x$, $R_y$ denote the side lengths of the user and BS HMIMO surfaces, respectively. 
Both HMIMO surfaces are perpendicular to the $z$-axis as shown in Fig.~\ref{figure_model}, and the distance between their center points along the $z$-axis is $d$.

From (\ref{spatial_frequency_x_y}) and (\ref{local_spatial_bandwidth}), assuming that the center-point of $\mathcal{S}$ on the $xy$-plane is at $(0,0)$, the local spatial bandwidth evaluated at the point $\bm{r}$ can be calculated as 
\begin{flalign}
    B(\bm{r},& \mathcal{S}) \!=\! \int_{\mathcal{A}\left( \kappa(\bm{r}, \bm{s}) \right)_{\bm{s} \in \mathcal{S}}} {\rm d}\kappa_x {\rm d}\kappa_y  \label{spatial_bandwidth_approx_2_1} \\
    \!\overset{(a)}{=}\!& \int_{-\frac{S_x}{2}}^{\frac{S_x}{2}} \! \int_{-\frac{S_y}{2}}^{\frac{S_y}{2}}
    \! \frac{\kappa_0^2 d^2}{\left(d^2\!+\!(r_x\!-\!s_x)^2\!+\!(r_y\!-\!s_y)^2 \right)^2} {\rm d}s_x {\rm d}s_y,  \label{spatial_bandwidth_approx_2_2}
\end{flalign}%
where $\mathcal{A}\left(\kappa\left(\bm{r},\bm{s}\right)\right)_{\bm{s} \in \mathcal{S}}$ represents the region covered by $\kappa\left(\bm{r},\bm{s}\right)$ in the wavenumber domain when $\bm{s}$ varies within $\mathcal{S}$. The derivation of the transformation in (a) is provided in Appendix \ref{appendix1}. Thus, assuming that the center-point of $\mathcal{R}$ on the $xy$-plane is at $(0,h_0)$ and according to (\ref{DoF}) and (\ref{spatial_bandwidth_approx_2_2}), the number of effective DoF can be formulated as 
\begin{align}
    \label{DoF_exact_integral}
    N \!\!=\! \!\frac{1}{\!(2\pi)^{\!2}\!} \!\!\int_{-\frac{R_x}{2}}^{\frac{R_x}{2}} \!\!\int_{-\frac{R_y}{2}+h_0}^{\frac{R_y}{2}+h_0} \!\int_{-\frac{S_x}{2}}^{\frac{S_x}{2}} \!\!\int_{-\frac{S_y}{2}}^{\frac{S_y}{2}}\!
    \frac{\kappa_0^2 d^2 {\rm d}s_x {\rm d}s_y{\rm d}r_x {\rm d}r_y }{\!\left(d^2\!\!+\!(r_x\!\!-\!s_x)^{\!2}\!+\!(r_y\!\!-\!s_y)^{\!2} \right)^{\!2}\!} 
\end{align}
However, the obtained quadruple integral does not admit a closed-form expression. To obtain a closed-form expression, we assume that $S_x, S_y \ll d$, since the side of the HMIMO surface of the user is usually much smaller than the propagation distance in many wireless scenarios. Then, the effective DoF can be approximated as 
\begin{flalign}
    \label{DoF_one_user}
    N & \!\!\approx \! \frac{\!S_x S_y\!}{\lambda^2} \!\Bigg\{\! \frac{R_x}{\!\!\sqrt{\!4d^2\!\!+\!\!R_x^2}} \!\!\left[ \!\arctan\!\!\left(\!\frac{R_y\!+\!2h_0}{\!\!\sqrt{\!4d^2\!\!+\!\!R_x^2}} \! \right) \!\!+\! \arctan\!\! \left(\!\frac{R_y\!-\!2h_0}{\!\!\sqrt{\!4d^2\!\!+\!\!R_x^2\!}} \! \right) \! \right]   \nonumber \\
    +& \frac{R_y\!+\!2h_0}{\!\sqrt{\!4d^2\!+\!(R_y\!+\!2h_0)^2}} \arctan \!\left(\!\frac{R_x}{\!\sqrt{\!4d^2\!+\!(R_y\!+\!2h_0)^2}} \!\right)  \nonumber \\
    +& \frac{R_y\!-\!2h_0}{\!\sqrt{\!4d^2\!+\!(R_y\!-\!2h_0)^2}} \arctan \!\left(\!\frac{R_x}{\!\sqrt{\!4d^2\!+\!(R_y\!-\!2h_0)^2}} \!\right)\! \Bigg\}.
\end{flalign}
The details of the derivation can be found in Appendix \ref{appendix2}. It is worth noting that (\ref{DoF_one_user}) coincides with \cite[Eq. (31)]{dardariCommunicatingLargeIntelligent2020} if $h_0=0$, as expected. In addition, we note that (\ref{DoF_exact_integral}) and (\ref{DoF_one_user}) are obtained without approximating the Lesbegue measure in (10), rather than by using, as done in \cite{dardariCommunicatingLargeIntelligent2020}, a quadrilateral to approximate it.

\begin{figure}[t]	
	\subfigure[] 
	{
		\begin{minipage}{4cm}
			\centering          
			\includegraphics[width=1.08\textwidth]{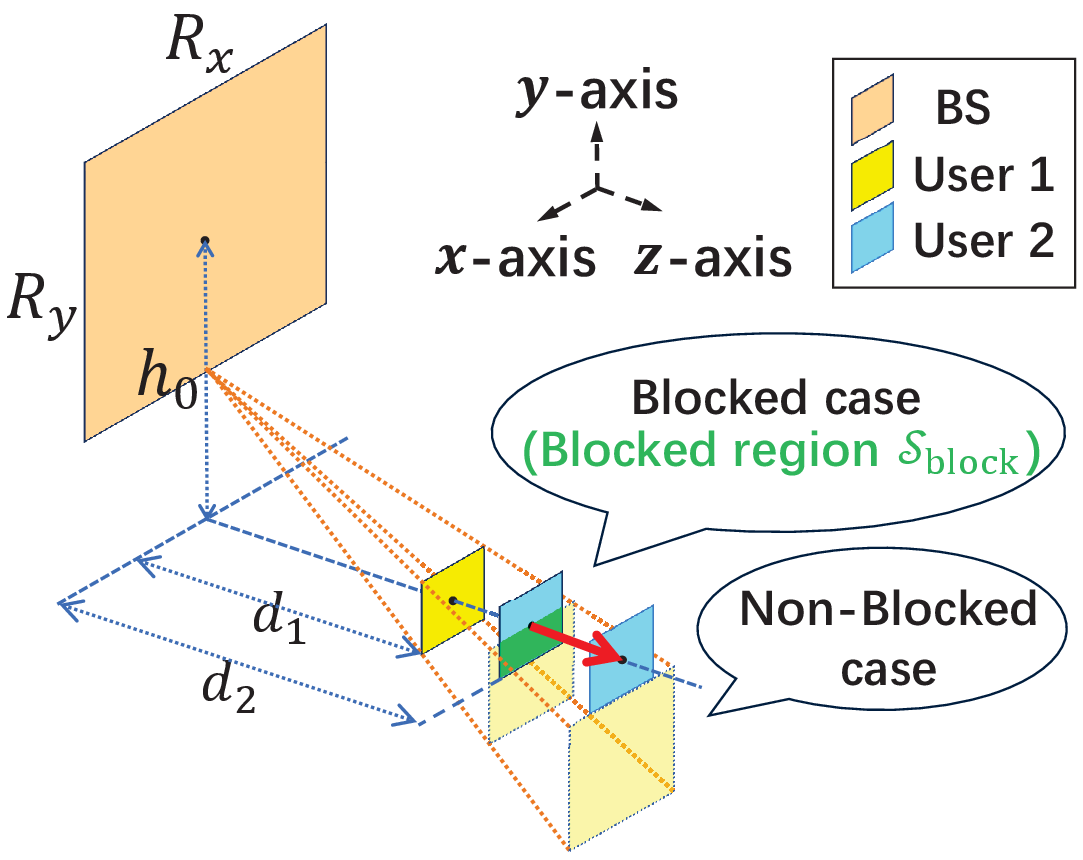}
      \label{fig:figure_blocked_and_not_scenario_1}
      \vspace{-0.8em}
		\end{minipage}
	}
  \hfill
	\subfigure[] 
	{
		\begin{minipage}{4cm}
			\centering      
			\includegraphics[width=1.08\textwidth]{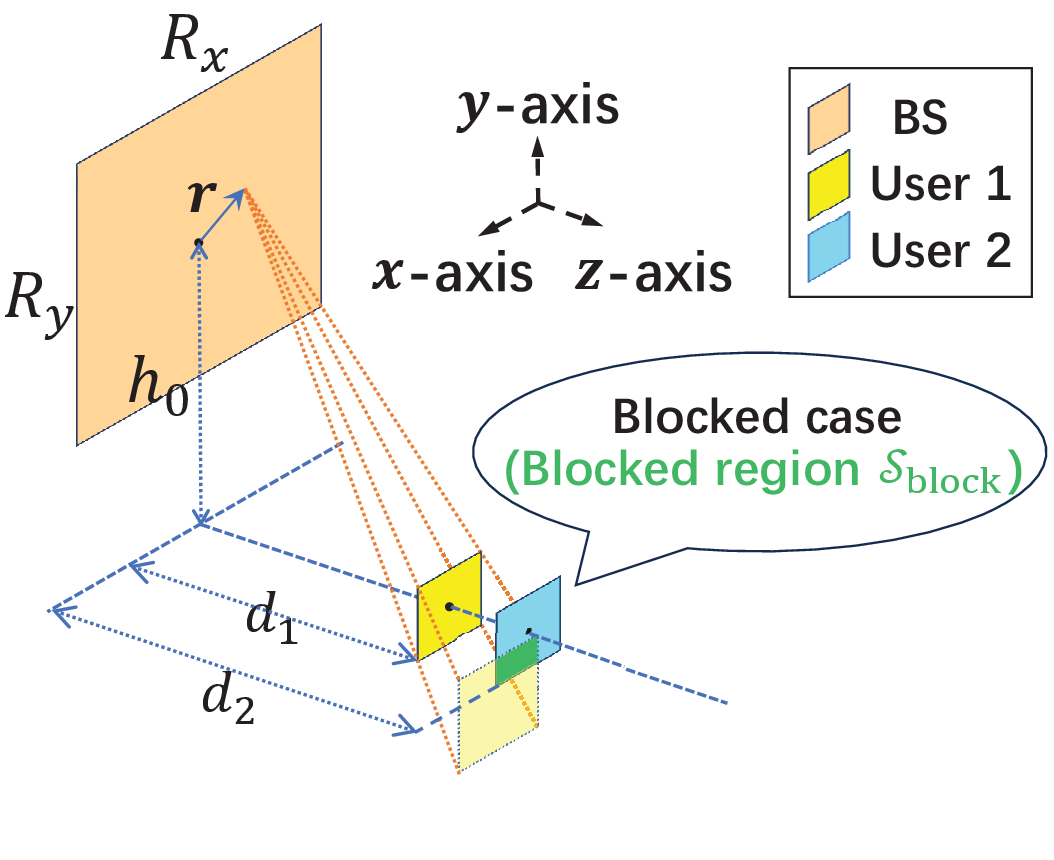} 
      \label{fig:figure_blocked_and_not_scenario_2}
      \vspace{-0.8em}
		\end{minipage}
	}
  \vspace{-0.8em}
	\caption{(a) Two case studies in which spatial blocking is not and is present; (b) The spatial blocking region for any $\bm{r}$ when spatial blocking cannot be ignored.} 
	\label{fig:figure_blocked_and_not_scenario}
  \vspace{-1.5em}
\end{figure}

\subsection{Multiple User Setup}
When multiple users communicate with the BS simultaneously, the region in the wavenumber domain spanned by the local spatial frequency corresponding to the HMIMO surfaces of different users, when evaluated at the same point $\bm{r}$, may overlap. In this case, the number of effective DoF is not equal to the sum of the number of effective DoF of each user, which is obtained when each user is the only one available in the network, i.e., the single-user case. The overlap in the wavenumber domain needs to be considered, which occurs when the link of a generic user and the BS is spatially blocked by any  other users' HMIMO surfaces. In this section, we analyze the case studies in which the spatial blocking is not and is present (as shown in Fig. \ref{fig:figure_blocked_and_not_scenario} for some case studies) in Proposition \ref{theorem1} and Proposition \ref{prop3}, respectively. For simplicity, the HMIMO surfaces of all the users have the same size.

\begin{prop}
    \label{theorem1}
    When the HMIMO surfaces of multiple users do not block each other at any observing position $\bm{r}$ on the HMIMO surface of the BS, the number of effective DoF is equal to the sum of the number of effective DoF for all the users, each one calculated in the single-user setting.
\end{prop}
\begin{proof}
    In mathematical terms, the absence of spatial blocking among the HMIMO surfaces of the users is the following:
{\begin{align}
    \label{proof_theorem1_1}
    \mathcal{A}\left(\kappa\left(\bm{r},\bm{s}\right)\right)_{\bm{s} \in \mathcal{S}_i} \cap \mathcal{A}\left(\kappa\left(\bm{r},\bm{s}\right)\right)_{\bm{s} \in \mathcal{S}_j} = \varnothing , i \neq j.
\end{align}%
}%
By virtue of the properties of the Lesbegue measure applied to disjoint sets, from (\ref{local_spatial_bandwidth}) we obtain the following: 
\begin{align}
    \label{proof_theorem1_2}
    \!\! B(\bm{r}, \mathcal{S}) \!=\! {\rm m}\!\left(\bm{\kappa} (\bm{r}, \bm{s})\right)_{\bm{s} \in \mathcal{S}_1 \cup \cdots \cup \mathcal{S}_L} 
    \!=\!\! \sum_{i=1}^{L} {\rm m}\!\left(\bm{\kappa} (\bm{r}, \bm{s})\right)_{\bm{s} \in \mathcal{S}_i},
\end{align}%
where $L$ denotes the number of users, and $\mathcal{S}_1, \!\cdots\!, \mathcal{S}_L$ represent the HMIMO surfaces of the users. Then, the number of effective DoF is equal to the sum of the number of effective DoF of the users, each one evaluated in the single-user case.
\end{proof}
Based on Proposition \ref{theorem1}, we aim to identify the geometric conditions under which the spatial blocking between two generic users occurs.
To this end, we consider the setup in Fig.~\ref{figure_model}, and consider User 1 and User 2 as a reference to derive the absence of spatial blocking between them. Based on Fig. 2, the distances from User 1 and User 2 to the BS (along the $z$-axis) are denoted by $d_1$ and $d_2$, respectively.

\begin{prop}\label{corollary1}
In the considered scenario, there is no spatial blocking between the links of the two users and the BS when the height $h_0$ of the HMIMO surface of the BS and the distance between the two users $\Delta d = d_2-d_1$ satisfy the conditions: 
\begin{align}
  \label{corollary1_1}
    h_0 \geq \frac{R_y+S_y}{2}, \quad
        \Delta d \geq \left(\frac{2S_y}{2h_0-S_y-R_y}\right)d_1.
\end{align}%
\end{prop}%
\begin{proof}
    As shown in Fig. \ref{fig:figure_blocked_and_not_scenario_1}, the conditions can be obtained by determining when the projection of User 1's HMIMO surface onto the plane where User 2's HMIMO surface is located does not overlap with User 2's HMIMO surface.
\end{proof}

If the conditions in (\ref{corollary1_1}) is not fulfilled, (\ref{proof_theorem1_1}) is not fulfilled and (\ref{proof_theorem1_2}) cannot be utilized, i.e., the Lesbegue measure of the local spatial frequency in the multi-user case is not the sum of the Lesbegue measures in the single-user case. In Proposition \ref{prop3}, we analyze the case study in which $h_0 \!>\! \frac{(R_y+S_y)}{2}$ is fulfilled but the second inequality in (\ref{corollary1_1}) is not.

\begin{prop}\label{prop3}
    When $0 \leq \Delta d < \left(\frac{2S_y}{2h_0-S_y-R_y} \right)d_1$ and $h_0 \geq \frac{R_y+S_y}{2}$, the number of effective DoF for the two-user setting shown in Fig.~\ref{fig:figure_blocked_and_not_scenario_2}, can be calculated as
    \begin{align}
        \label{theorem2_1}
        N = N_1 + N_2 - N_{\rm {blocked}},
    \end{align}
    where $N_1$ and $N_2$ are the numbers of effective DoF of User 1 and User 2 in the single-user case, respectively, and $N_{\rm {blocked}}$ accounts for the spatial blocking between the HMIMO surfaces of User 1 and User 2, which can be calculated as
    \begin{align}
        \label{theorem2_2}
        N_{\rm {blocked}} = \frac{1}{(2\pi)^2} \int_\mathcal{R} B(\bm{r}, \mathcal{S}_{\rm {blocked}}(\bm{r}))_{\rm overlap} {\rm d}\bm{r}, 
        \\
        B(\bm{r}, \mathcal{S}_{\rm {blocked}}(\bm{r}))_{\rm {overlap}} = {\rm m}\left(\bm{\kappa} (\bm{r}, \bm{s})\right)_{\bm{s}\in \mathcal{S}_{\rm blocked}(\bm{r})},
    \end{align}
where $\mathcal{S}_{\rm {blocked}}(\bm{r})$ denotes the portion of User 2's surface that is blocked by User 1's surface, and $B(\bm{r}, \mathcal{S}_{\rm {blocked}}(\bm{r}))_{\rm {overlap}}$ denotes the corresponding local spatial bandwidth. The expression of $\mathcal{S}_{\rm {blocked}}(\bm{r})$ can be found in (\ref{m_s_blocked}) at the top of this page, where $\otimes$ denotes the Cartesian product of sets.
\end{prop}
\begin{proof}
    The proof follows directly from (\ref{proof_theorem1_1}) and (\ref{proof_theorem1_2}), by recalling the definition for the union of non-disjoint sets. Eq. (\ref{m_s_blocked}) is obtained using the geometric relationship between the HMIMO surfaces of User 1 and User 2. As shown in Fig. \ref{fig:figure_blocked_and_not_scenario_2}, the blocked region $\mathcal{S}_{\rm {blocked}}(\bm{r})$ can be obtained by projecting User 1's HMIMO surface onto User 2's HMIMO surface.
\end{proof}
\begin{figure*}[!hbpt]
\begin{equation}
      \label{m_s_blocked}
      \!\mathcal{S}_{\text{blocked}}(\bm{r}) \!=\!\!
      \begin{cases}
        \!\!\left[\!-\frac{S_x}{2}, \frac{S_x}{2}\!\right] \!\!\otimes\!\! \left[\!-\frac{S_y}{2}, -\frac{\Delta d}{d_1} (h \!+\! r_y) \!+\! \frac{S_yd_2}{2d_1} \right]\!, & \hspace{-4pt} \text{if} \hspace{2pt} (r_x, r_y) \!\in \!\!\left[ -\frac{S_x}{2}, \frac{S_x}{2} \right] \!\otimes\! \left[-\frac{R_y}{2}, -h_0 \!+\! \frac{S_y}{2} \!\left(\!  \frac{d_2\!+\!d_1}{d_2\!-\!d_1}\!\right) \right] \\[5pt]
          \!\!\left[\!-\frac{\!S_x}{2},\! -\frac{r_x \Delta d}{d_1} \!+\! \frac{S_x d_2}{2d_1}\! \right] \!\!\otimes\!\! \left[\!-\frac{\!S_y}{2}, \!-\frac{\Delta d}{d_1} (h \!+\! r_y) \!+\! \frac{S_y d_2}{2d_1} \!\right]\!, & \hspace{-4pt} \text{if} \hspace{2pt} (r_x, r_y) \!\in \!\!\left[\frac{S_x}{2}, \frac{S_x}{2} \!\left(\!  \frac{d_2\!+\!d_1}{d_2\!-\!d_1}\!\right) \right] \!\!\otimes\!\! \left[\!-\frac{R_y}{2}, -h_0 \!+\! \frac{S_y}{2} \!\left(\!  \frac{d_2\!+\!d_1}{d_2\!-\!d_1}\!\right)\!\right] \\[5pt]
          \!\!\left[\frac{r_x \Delta d}{\!d_1} \!-\! \frac{S_x d_2}{2d_1} \!, \frac{\!S_x}{\!2}\!\right] \!\!\otimes\!\! \left[\!-\frac{\!S_y}{\!2}, -\frac{\Delta d}{\!d_1}  (h \!+\! r_y) \!+\! \frac{S_y d_2}{2 d_1}\!\right]\!, & \hspace{-4pt}\text{if} \hspace{2pt} (r_x, r_y) \!\in \!\!\left[\!-\frac{\!S_x}{\!2} \!\left(\!  \frac{d_2\!+\!d_1}{d_2\!-\!d_1}\!\right)  \!,\! -\frac{\!S_x}{\!2}\!\right] \!\!\otimes \!\!\left[\!-\frac{\!R_y}{\!2}, \!-h_0 \!+\! \frac{\!S_y}{\!2} \!\left(\!  \frac{d_2\!+\!d_1}{d_2\!-\!d_1}\!\right)\!\right] \\[5pt]
          0 & \hspace{-4pt} \text{otherwise} 
      \end{cases}
      \vspace{-0.6em}
\end{equation}
\hrulefill
\vspace{-1.5em}
\end{figure*}

In conclusion, the number of effective DoF in the multi-user case can be determined by considering the spatial blocking between the users' HMIMO surfaces. These effects are quantified by identifying the regions on each user's HMIMO surface that are blocked by the surfaces of the other users. The number of effective DoF is obtained by subtracting the number of effective DoF that correspond to the blocked regions from the sum of the number of effective DoF for all users, each one calculated in the single-user setting.

\section{Simulation Results} \label{SIM}
In this section, simulation results are provided to validate the effectiveness of the proposed method to calculate the number of effective DoF in near-field channels.
The wavelength is set to $\lambda = 0.01$ m, and the sizes of the receiving and transmitting HMIMO surfaces are set to $R_x = R_y = 1.4$ m and $S_x = S_y = 0.3$ m, respectively.
The height of the BS surface is set to $h_0 = 5$ m, and the distance between the BS and User 1 is $d$.
Since the number of effective DoF sharply decreases with the increase of $d$, and it increases with the increase of $R_x$ and $R_y$, we illustrate the results as a function of the relative distance $F = (d^2)/(R_xR_y)$, with values ranging from $15$ dB to $30$ dB. 

To validate the proposed method, we adopt the singular value decomposition (SVD) method to estimate the number of effective DoF numerically. 
Each HMIMO surface is divided into small patches of size $\frac{\lambda}{3} \!\times\! \frac{\lambda}{3}$, where the current distribution is assumed constant~\cite{xie2023performance}. 
Then, we solve the eigenfunction problem by applying the SVD to the corresponding channel matrix. 
The number of effective DoF is calculated by counting the number of eigenvalues whose normalized (with respect to the largest eigenvalue) values are greater than 0.5~\cite{dardariCommunicatingLargeIntelligent2020, Renzo2023}.

Fig. \ref{fig:sim_for_single_user} shows the number of effective DoF in the single-user setting as a function of $F$ and for different values of $S_x S_y$.
It can be observed that the theoretical analysis in (\ref{DoF_one_user}) is consistent with the SVD. 
As the distance between the BS and the user increases (maintaining the sizes of the HMIMO surfaces), the number of DoF decreases sharply. 
This is because the local spatial bandwidth decreases with the increase of the distance,
which results in decreasing the number of orthogonal sub-channels.
When the distance approaches infinity, the effective number of DoF approaches one, as expected in far-field channels.

\begin{figure}
    \setlength{\abovecaptionskip}{-0.3em}
    \centering
    \includegraphics[width=0.7\linewidth]{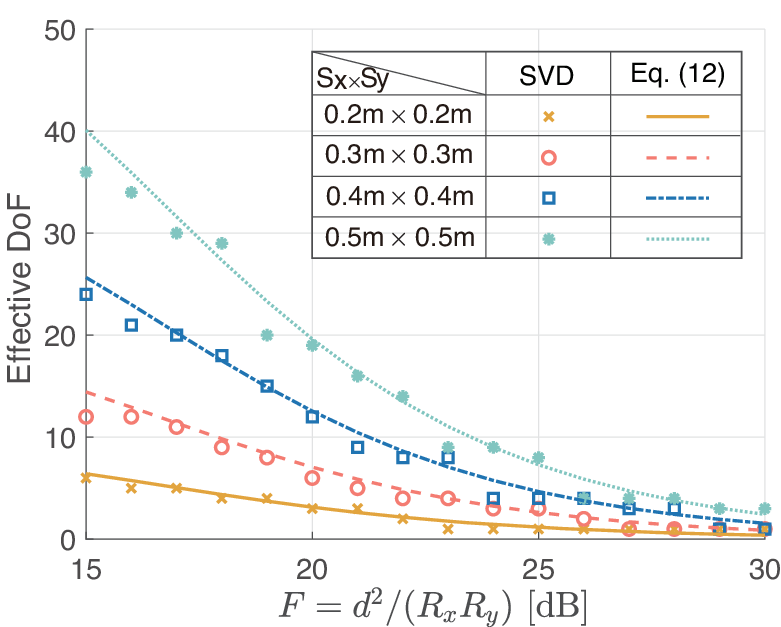}
    \caption{Number of effective DoF - Single user case.}
    \label{fig:sim_for_single_user}
    \vspace{-1.5em}
\end{figure}
\begin{figure}
    \setlength{\abovecaptionskip}{-0.3em}
    \centering 
    \includegraphics[width=0.71\linewidth]{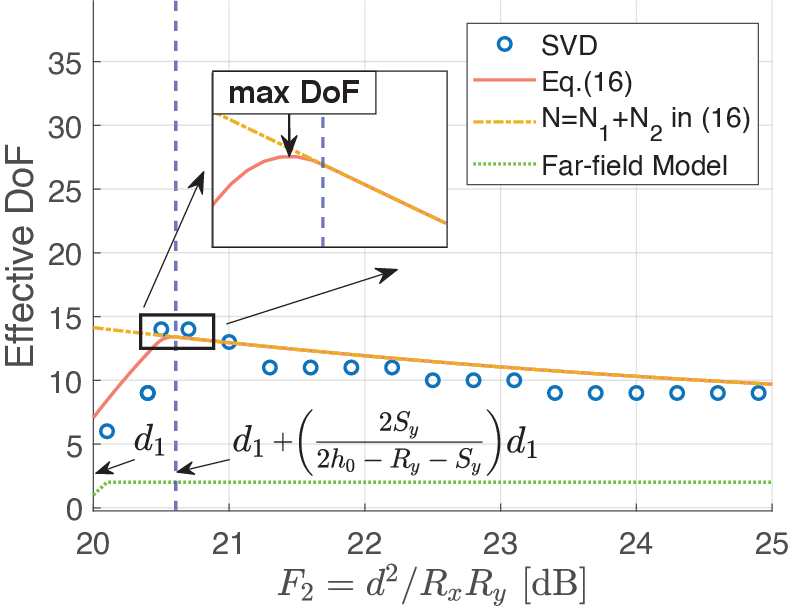}
    \caption{Number of effective DoF - Two users case.}
    \label{fig:sim_for_two_users}
    \vspace{-1.5em}
\end{figure}

Fig. \ref{fig:sim_for_two_users} shows the number of effective DoF in the two-user setting as a function of the distance $d_2$, given $d_1$.
The distance, $d_1$, between the BS and User 1 is set to $F_1 = d_1^2/(R_x R_y)=20 {\rm dB}$, and the distance, $d_2$, between the BS and User 2 is set so that $F_2 = d_2^2/(R_x R_y)$ is varied from 20 dB to 25 dB.
The number of effective DoF based on the far-field plane-wave model is illustrated as a reference as well.
The results show the good accuracy of the proposed analytical framework compared with the SVD.
When $F_2$ is small, the gap between the sum of the number of effective DoF for all the users and the SVD estimate is large, indicating that the impact of spatial blocking is significant in the considered setup.
As $F_2$ increases, the sum of the number of effective DoF for all the users gradually approaches the SVD estimate, since the blocking between the HMIMO surfaces of the users becomes negligible.
When the distance between the two users is larger than $ (\frac{2S_y}{2h_0-S_y-R_y})d_1$, the number of effective DoF can be obtained by directly summing the number of effective DoF of all the users, 
which is consistent with the theoretical analysis in Section \ref{section3}.
In the considered setup, the average reduction in the number of effective DoF due to spatial blocking is $21.2\%$.
Moreover, we observe that the number of effective DoF initially increases and then decreases as the distance between the two users increases. 
This suggests an optimal distance at which the system achieves the maximum number of effective DoF. 
It is interesting to note that the maximum number of effective DoF occurs when User 2's surface is partially blocked by User 1's surface.
The reason is that the increase in the number of effective DoF due to the reduced distance offsets the loss in the number of effective DoF caused by the blocking. 

\section{Conclusion} \label{Conclusion}
In this paper, we have proposed an analytical framework to calculate the number of effective DoF in near-field channels and in the presence of multiple HMIMO users.
The proposed method is based on the computation of the spatial bandwdith of the received electromagnetic field. Specifically, the impact of spatial blocking between two HMIMO surfaces has been taken into account, which is proved to decrease the number of effective DoF. 
In addition, numerical results have highlighted that the presence of spatial blocking has a non-negligible impact on the number of effective DoF and that there exists an optimal distance between the users that maximizes the number of effective DoF, which depends on the trade-off between transmission distance and spatial blocking.

\begin{appendices} 
\section{Derivation of (\ref{spatial_bandwidth_approx_2_2})} \label{appendix1}
According to \eqref{spatial_frequency_x_y}, when evaluated at the observation point $\bm{r}$, ${\rm d}\kappa_x {\rm d}\kappa_y$ can be calculated using the coordinate transformation: 
\vspace{-1em}
\begin{align}
    \label{coordinate_transformation}
    {\rm d}\kappa_x {\rm d}\kappa_y = \left\vert J \right\vert {\rm d}s_x {\rm d}s_y,
\end{align}
where $\left\vert J \right\vert$ denotes the Jacobian determinant of the mapping $(s_x, s_y) \!\rightarrow\! (\kappa_x, \kappa_y)$. From \eqref{spatial_frequency_x_y}, $\left\vert J \right\vert$ can be expressed as 
\begin{align}
    \label{Jacobi_determinant}
    \left\vert J \right\vert = \left \vert \begin{matrix}
        \frac{\partial \kappa_x}{\partial s_x} &  \frac{\partial \kappa_x}{\partial s_y}\\
          \frac{\partial \kappa_y}{\partial s_x}& \frac{\partial \kappa_y}{\partial s_y}
    \end{matrix} \right \vert
    \!=\! \frac{\kappa_0^2 d^2}{\left(d^2 + (r_x - s_x)^2 + (r_y - s_y)^2\right)^{2}}.
\end{align}
By combining (\ref{coordinate_transformation}) and (\ref{Jacobi_determinant}), the local spatial bandwidth $B(\bm{r}, \mathcal{S})$ in (\ref{spatial_bandwidth_approx_2_2}) is obtained.

\section{Derivation of (\ref{DoF_one_user})} \label{appendix2}
According to (\ref{DoF_exact_integral}), when $S_x, S_y \ll d$, we obtain
\begin{align}
    \label{DoF_approx_integral}
    N \approx & \frac{S_x S_y}{\lambda^2} \int_{-\frac{R_x}{2}}^{\frac{R_x}{2}} \int_{-\frac{R_y}{2}+h_0}^{\frac{R_y}{2}+h_0} 
    \frac{d^2}{\left(d^2+r_x^2+r_y^2\right)^2} {\rm d}r_x {\rm d}r_y,
\end{align}
where the approximation is obtained by applying the $0$th-order Taylor approximation of the integrand function. 
By calculating the integral in (\ref{DoF_approx_integral}), the expression in (\ref{DoF_one_user}) is obtained.

\end{appendices}

\bibliographystyle{ieeetr}
\bibliography{DoF}

\vfill

\end{document}